\documentclass[a4paper]{amsart}

\usepackage{amsmath,amssymb,amsthm}
\usepackage[english]{babel}
\usepackage{paralist}

\usepackage{graphicx}

\newcommand{\abs}[1]{\left\vert #1\right\vert}
\newcommand{\amax}{a_\textup{max}}
\renewcommand{\d}{\textup{d}}
\newcommand{\Lip}[1]{\operatorname{Lip}(#1)}
\newcommand{\cM}{\mathcal{M}_+^N(\R^2)}
\newcommand{\norm}[2]{\left\Vert #1\right\Vert_{#2}}
\newcommand{\R}{\mathbb{R}}
\newcommand{\unit}[1]{~\mathrm{#1}}
\newcommand{\vd}{v_\d}
\newcommand{\vmax}{v_\textup{max}}

\theoremstyle{remark}\newtheorem*{remark}{Remark}
\theoremstyle{plain}
	 \newtheorem{proposition}{Proposition}[section]
	 \newtheorem{theorem}[proposition]{Theorem}
\theoremstyle{definition}\newtheorem{assumption}[proposition]{Assumption}

\graphicspath{{./}{figures/}}

\title{Moving in a crowd: human perception as a multiscale process}
\author{Annachiara Colombi}
\address{Department of Mathematical Sciences ``G. L. Lagrange'', Politecnico di Torino, Corso Duca degli Abruzzi 24, 10129 Torino, Italy}
\email{annachiara.colombi@polito.it}
\author{Marco Scianna}
\address{Department of Mathematical Sciences ``G. L. Lagrange'', Politecnico di Torino, Corso Duca degli Abruzzi 24, 10129 Torino, Italy}
\email{marco.scianna@polito.it}
\author{Andrea Tosin}
\address{Department of Mathematical Sciences ``G. L. Lagrange'', Politecnico di Torino, Corso Duca degli Abruzzi 24, 10129 Torino, Italy}
\email{andrea.tosin@polito.it}

\begin{document}


\begin{abstract}
The strategic behaviour of pedestrians is largely determined by how they perceive and react to neighbouring people. This issue is addressed in this paper by a model which combines, in a time and space-dependent way, discrete and continuous effects of pedestrian interactions. Numerical simulations and qualitative analysis suggest that human perception, and its impact on crowd dynamics, can be effectively modelled as a multiscale process based on a dual microscopic/macroscopic representation of groups of agents.
\end{abstract}

\maketitle

\section{Introduction and motivations}
In this paper we aim at incorporating the effect of pedestrian perception in a mathematical description of interpersonal interactions.

We take inspiration from~\cite{fujiyama2005WP,fujiyama2006WP}, where the author points out that different perceptions of the surroundings can lead walkers to react in a more individualistic or group-oriented way to the presence of nearby people. In particular, he introduces the concept of the \emph{use of space} as an indicator of such a behaviour, implying that this affects the pedestrian collision avoidance mechanism.

We start by the celebrated \emph{social force model}~\cite{helbing1995PRE} in the simple case of a single moving walker and we enrich it by introducing a multiscale micro/macroscopic representation of a group of individuals composing a static crowd that the walker interacts with. For this, we take advantage of the measure-theoretic multiscale approach developed in~\cite{cristiani2014BOOK}. The multiscale representation uses a \emph{perception function}, to be related to the aforesaid use of space, which determines how much the interactions of the walker are directed towards either the individual (viz. microscopic) or the collective (viz. macroscopic) distribution of the nearby people.

Our results show that such a multiscale interpretation of the effect of human perception can greatly impact on the correct reproduction of  pedestrian trajectories and that this may not be equally possible with a single-scale model.

\section{Mathematical model}
We consider for simplicity a single pedestrian in a two-dimensional domain, who walks through a static crowd to reach a given target. The pedestrian is represented by his/her position and velocity $x(t),\,v(t)\in\R^2$, respectively, where $t\geq 0$ is time. The distribution of the static individuals is instead described by a Radon positive measure $\mu$ carrying a total mass $\mu(\R^2)=N$, i.e., the number of static individuals.

The dynamics of the walker are ruled by a social-force-type model~\cite{helbing1995PRE}:
\begin{subequations}
\begin{align}
	& \dot{x}(t)=v(t)=g\left(\frac{\abs{w(t)}}{\vmax}\right)w(t) \label{eq:xdot} \\
 	& \dot{w}(t)=\frac{\vd(x(t))-v(t)}{\tau}+\int_{S_R^\alpha(x(t))}K(y-x(t))\,d\mu(y), \label{eq:wdot}
\end{align}
\end{subequations}
where $g(z)=\min\{1,\,\frac{1}{z}\}$ bounds the actual speed $\abs{v(t)}$ by a physiological maximal value $\vmax>0$.

In~\eqref{eq:wdot}, $\vd:\R^2\to\R^2$ is a given \emph{desired velocity} representing the preferred direction of the moving pedestrian to reach his/her destination from his/her current position, and $\tau$ is a relaxation time. The second term at the right-hand side models instead the interactions with the static individuals. In particular, it expresses the tendency to keep a sufficient distance from them for collision avoidance. The \emph{interaction kernel} $K:\R^2\to\R^2$ describes the position-dependent repulsion of the walker from the static individuals within in his/her \emph{sensory region} $S_R^\alpha(x(t))$ (see Fig.~\ref{fig:setup}, bottom-left panel).

\subsection{Modelling perception: the multiscale structure of $\mu$}
The measure $\mu$ is used to describe how the static individuals are perceived by the moving pedestrian, who can interact with them either singularly or group-wise depending on his/her use of space (see Fig.~\ref{fig:setup} top-left panel).

Taking inspiration from~\cite{colombi2015JMB,cristiani2014BOOK}, we assume that a highly localised perception, typical of relaxed conditions, induces a quite accurate use of space, hence individualistic interactions. In this case we choose $\mu$ as an atomic mass measure $\mu=\epsilon:=\sum_{k=1}^{N}\delta_{y_k}$, where $\delta$ is the Dirac delta and $y_k\in\R^2$ is the position of the $k$th static individual. Conversely, a blurred perception, typical of hurried or panicky conditions, induces a rougher assessment of the usable space, hence group-oriented interactions. In this case we choose $\mu$ as a Lebesgue-absolutely continuous measure $\mu=\rho\mathcal{L}^2$, where $\rho:\R^2\to [0,\,+\infty)$ is the density of the static crowd. With a slight abuse of notation, we will denote the measure $\mu$ by the same symbol $\rho$ and we will require that $\rho(\R^2)=\int_{\R^2}\rho(y)\,dy=N$.

The moving pedestrian can also change his/her type of perception while walking, for instance according to local characteristics of the static crowd. We model this by generalising $\mu$ as
\begin{equation*}
	\mu_t=\theta(x(t))\epsilon+\bigl(1-\theta(x(t))\bigr)\rho,
	\label{eq:mu_hybr}
\end{equation*}
where $\theta:\R^2\to [0,\,1]$ is the \emph{level of perception}. $\theta\searrow 0$ indicates a blurred perception, with the moving pedestrian tending to assess the space occupancy in a continuous way. Conversely, $\theta\nearrow 1$ indicates a localised perception, with the moving pedestrian tending to assess it in a discrete way. Note that the dependence of $\theta$ on $x(t)$ makes the measure $\mu$ time-dependent.


\section{Numerical simulations}
\begin{figure}[t!]
\centering
\includegraphics[width=\textwidth]{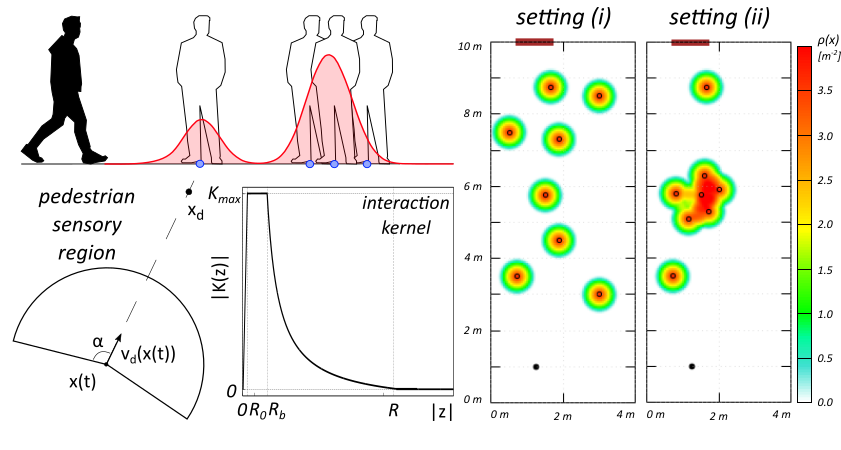}
\caption{Top-left: pictorial representation of the multiscale perception. Bottom-left: sensory region and interaction kernel. Right: specification of settings (i), (ii)}
\label{fig:setup}
\end{figure}

We perform numerical simulations of model~\eqref{eq:xdot}-\eqref{eq:wdot} in a two-dimensional rectangular domain of size $4\unit{m}\times 10\unit{m}$, which is meant to reproduce a corridor or a pavement.

The moving pedestrian, initially in $x(0)=(1.25,\,1)\unit{m}$, wants to reach a $85\unit{cm}$-wide target on the top edge of the domain. In doing this, s/he faces $N=8$ static individuals. We consider two spatial arrangements of the latter:
\begin{inparaenum}[(i)]
\item they are sparse;
\item they form a dense cluster in the central part of the domain,
\end{inparaenum}
see Fig.~\ref{fig:setup} right panels. Given their microscopic positions $\{y_k\}_{k=1}^{N}$, we construct their macroscopic density $\rho$ as the superposition of $N$ unit-mass cones:
\begin{equation}
	\rho(y)=\frac{3}{\pi\sigma^2}\sum_{k=1}^N\left(1-\frac{\abs{y-y_k}}{\sigma}\right)\chi_{B_\sigma(y_k)}(y)
	\label{eq:rho}
\end{equation}
with $\sigma=0.5\unit{m}$, $\chi_{B_{\sigma}(y_k)}$ being the characteristic function of the ball centred in $y_k$ with radius $\sigma$.

\begin{table}[t]
\caption{Summary of the parameters used in the model}
\label{tab:parameters}
\begin{tabular}{lllc}
\hline\hline\noalign{\smallskip}
Parameter & Description & Value & Reference \\
\hline
$\vmax$ & pedestrian maximum speed & $1.34\unit{m/s}$ & \cite{helbing2009ENCYCLOPEDIA} \\
$\tau$ & relaxation time & $0.5\unit{s}$ & \cite{helbing2009ENCYCLOPEDIA} \\
$R$ & sensory radius & $3\unit{m}$ & \cite{helbing2009ENCYCLOPEDIA} \\
$\alpha$ & half visual angle & $100^\circ$ & \cite{helbing2009ENCYCLOPEDIA} \\
$k_0$ & interpersonal repulsion coefficient & $0.3\unit{m^2/s^2}$ & tuned here \\
$R_b$ & average pedestrian body radius & $0.3\unit{m}$ & \cite{seyfried2005JSM,venuti2007CRM} \\
\hline
\end{tabular}
\end{table}

System~\eqref{eq:xdot}-\eqref{eq:wdot} requires the specification of some parameters, see Table~\ref{tab:parameters}. Moreover, we define $\vd(x)=\vmax\frac{x_\d-x(t)}{\abs{x_\d-x(t)}}$, where $x_\d=(1.2,\,10)\unit{m}$ is the centre of the target. We set the sensory region of the moving pedestrian to be the circular sector
$$ S_R^\alpha(x(t))=\left\{y\in\R^2:\abs{y-x(t)}\leq R,\ \frac{(y-x(t))\cdot\vd(x(t))}{\vmax\abs{y-x(t)}}\geq\cos{\alpha}\right\}, $$
where $R$ is the interaction radius and $\alpha$ is the half visual angle. This circular sector is oriented in such a way that the gaze direction of the moving pedestrian is aligned with $\vd$, thus with the target (cf. Fig.~\ref{fig:setup}, bottom-left panel). Finally, we take the interaction kernel as a classical distance-decaying function:
$$ K(r)=
	\begin{cases}
		-k_0\left(\frac{1}{R_b}-\frac{1}{R}\right)\frac{r}{\abs{r}} & \text{if } 0\leq\abs{r}\leq R_b \\
		-k_0\left(\frac{1}{\abs{r}}-\frac{1}{R}\right)\frac{r}{\abs{r}} & \text{if } R_b<\abs{r}\leq R \\
		0 & \text{otherwise}
	\end{cases}
$$
(cf. Fig.~\ref{fig:setup} bottom-left panel), where $R_b<R$ is the body size of an average individual and $k_0>0$ is a proportionality coefficient.

We now perform numerical tests to see how different types of perception give rise to different migratory paths of the moving pedestrian. We consider either a fully localised perception, given by $\theta\equiv 1$, which corresponds to the genuinely microscopic social-force-type model, or a hybrid one. In this latter case, we assume that the walker has a localised perception when the static individuals in $S_R^\alpha(x(t))$ are sparse enough. On the contrary, when they are more densely packed s/he perceives them as an undifferentiated group. The discriminating quantity is the mean distance $\ell$ among the static individuals within the sensory region:
\begin{equation}
	\ell=\ell(x(t))=\frac{1}{n(n-1)}\sum_{y_h,y_k\in S_R^\alpha(x(t))}\abs{y_h-y_k},
	\label{eq:ell}
\end{equation}
where $n=\#\{y_k\in S_R^\alpha(x(t)),\ k=1,\,\dots,\,N\}$ is their number. Then we set:
\begin{equation}
	\theta=\theta\left(\frac{\ell}{\ell^\ast}\right)=
		\begin{cases}
			0 & \text{if } 0\leq\frac{\ell}{\ell^\ast}\leq 1 \\
			\frac{\ell}{\ell^\ast}-1 & \text{if } 1<\frac{\ell}{\ell^\ast}\leq 2 \\
			1 & \text{if } \frac{\ell}{\ell^\ast}>2
		\end{cases}
	\label{eq:theta}
\end{equation}
where $\ell^\ast=1\unit{m}$ is a reference value. Actually,~\eqref{eq:ell} is valid only if $n\geq 2$. If instead $n=0,\,1$ we invariably set $\theta=1$.

\begin{figure}[!t]
\centering
\includegraphics[width=\textwidth]{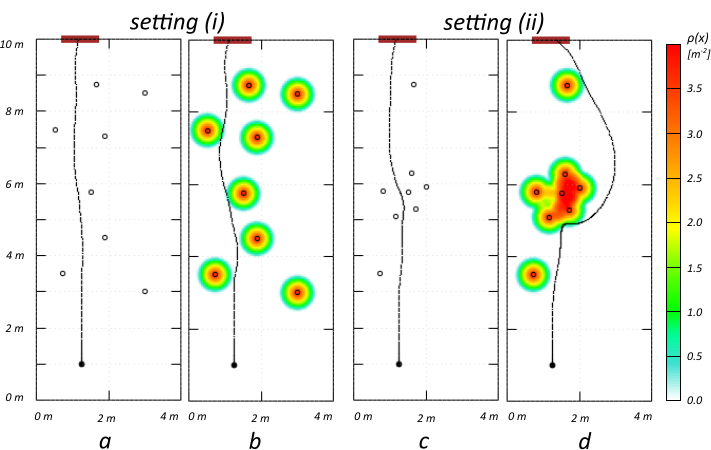}
\caption{Paths followed by the moving pedestrian in the two simulation settings}
\label{fig:paths}
\end{figure}

As shown in Fig.~\ref{fig:paths}a,~c, in both settings (i) and (ii) the fully localised perception allows the walker to pass in between the static individuals, thereby following an almost straight path towards the target. A hybrid variable perception results instead in different trajectories depending on the crowd distribution. When the static individuals are sufficiently sparse (setting (i), Fig.~\ref{fig:paths}b) the moving pedestrian still perceives them as a set of single elements and s/he uses the free space among them. Conversely, when they are more densely packed (setting (ii), Fig.~\ref{fig:paths}d) the moving pedestrian perceives predominantly their ensemble as a compact distributed mass and s/he circumnavigates the density spot. Note that, instead, the purely microscopic model may not allow one to appreciate substantial differences between the migratory paths in settings (i) and (ii) (cf. Figs.~\ref{fig:paths}a, c).

\section{Analysis of the trajectories}
We now study the dependence of the trajectory $t\mapsto x(t)$ on the perception function $\theta$ and on the multiscale description $(\epsilon,\,\rho)$ of the static crowd. We begin by rewriting~\eqref{eq:wdot} in the compact form $\dot{w}(t)=a[\mu_t](x(t),\,w(t))$, where $a$ stands for the acceleration of the moving pedestrian.
\begin{assumption}
We assume that $a$ is bounded and Lipschitz continuous, i.e., there exist $\amax,\,\Lip{a}>0$ s.t.:
\begin{align*}
	& \abs{a[\mu](x,\,w)}\leq\amax \\
	& \abs{a[\nu](x_2,\,w_2)-a[\mu](x_1,\,w_1)}\leq\Lip{a}\left(\abs{x_2-x_1}+\abs{w_2-w_1}+W_1(\mu,\,\nu)\right)
\end{align*}
for all $x,\,x_1,\,x_2,\,w,\,w_1,\,w_2\in\R^2$ and $\mu,\,\nu\in\cM$.
\label{ass:a}
\end{assumption}

\begin{remark}
Here and henceforth $\cM$ is the cone of positive measures with mass $N$ in $\R^2$. Furthermore, $W_1$ is the first Wasserstein metric in the space of finite positive measures.
\end{remark}

Next we consider any two measures $\mu^1_t,\,\mu^2_t\in\cM$ describing the distribution of the static crowd and we let $(x_1(t),\,w_1(t))$, $(x_2(t),\,w_2(t))$ be the corresponding trajectory-velocity pairs of the moving pedestrian.
\begin{proposition}
Let $x_1(0)=x_2(0)$ and $w_1(0)=w_2(0)$. There exists a constant $C>0$ such that
\begin{equation}
	\abs{x_2(t)-x_1(t)}\leq Ce^{Ct}\int_{0}^{t}W_1(\mu^1_s,\,\mu^2_s)\,ds
	\label{eq:Gronwall_x}
\end{equation}
for all $0\leq t\leq T<+\infty$.
\label{prop:Gronwall_x}
\end{proposition}
\begin{proof}
Integrating the acceleration in time in the two cases and taking the difference gives 
$$ \abs{w_2(t)-w_1(t)}\leq\int_0^t\abs{a[\mu^2_s](x_2(s),\,w_2(s))-a[\mu^1_s](x_1(s),\,w_1(s))}\,ds, $$
whence, by Assumption~\ref{ass:a} and Gronwall's inequality,
\begin{equation}
	\abs{w_2(t)-w_1(t)}\leq\Lip{a}e^{\Lip{a}t}\int_0^t\left(\abs{x_2(s)-x_1(s)}+W_1(\mu^1_s,\,\mu^2_s)\right)\,ds.
	\label{eq:Gronwall_w}
\end{equation}
Now, integrating~\eqref{eq:xdot} in time and using the boundedness and Lipschitz continuity of $g$ we obtain
\begin{multline*}
	\abs{x_2(t)-x_1(t)}\leq\frac{\amax\Lip{g}t}{\vmax}\int_0^t\abs{w_2(s)-w_1(s)}\,ds \\
		+\Lip{a}t\int_0^t\left(\abs{x_2(s)-x_1(s)}+\abs{w_2(s)-w_1(s)}+W_1(\mu^1_s,\,\mu^2_s)\right)\,ds,
\end{multline*}
which, invoking~\eqref{eq:Gronwall_w}, after standard manipulations produces
$$ \abs{x_2(t)-x_1(t)}\leq\alpha(t)\left(\int_0^t\abs{x_2(s)-x_1(s)}\,ds+\int_0^t W_1(\mu^1_s,\,\mu^2_s)\,ds\right), $$
with $\alpha(t)=\Lip{a}t\left[1+\left(\frac{\amax\Lip{g}}{\vmax}+\Lip{a}\right)te^{\Lip{a}t}\right]$. Since $\alpha(t)$ is non-decre\-asing, we set $C=\alpha(T)\geq\alpha(t)$ and by Gronwall's inequality we get the thesis.
\end{proof}

It is not difficult to check that slightly regularised versions of both $\vd$, cf.~\cite{cristiani2014BOOK}, and the acceleration in~\eqref{eq:wdot} satisfy Assumption~\ref{ass:a}. In particular, we propose
$$ a[\mu](x,\,w)=\frac{\vd(x)-g\left(\frac{\abs{w}}{\vmax}\right)w}{\tau}+\int_{\R^2}K(y-x)\eta_{S_R^\alpha(x)}(y)\,d\mu(y), $$
where $\eta_{S_R^\alpha(x)}:\R^2\to [0,\,1]$ is a mollification of the characteristic function of the set $S_R^\alpha(x)$. To see that Assumption~\ref{ass:a} is satisfied, use the boundedness and Lipschitz continuity of $\vd$ and $g$ and the results contained in~\cite{tosin2011NHM}.

Thanks to Proposition~\ref{prop:Gronwall_x} we are now in a position to prove
\begin{theorem}
Let $\theta_1,\,\theta_2:\R^2\to [0,\,1]$ be Lipschitz continuous and $\mu^i_t=\theta_i(x_i(t))\epsilon+\bigl(1-\theta_i(x_i(t))\bigr)\rho$, $i=1,\,2$, the corresponding multiscale measures. There exists $C>0$, which depends on $e^{\min\{\Lip{\theta_1},\,\Lip{\theta_2}\}W_1(\epsilon,\,\rho)}$, such that
$$ \sup_{t\in [0,\,T]}\abs{x_2(t)-x_1(t)}\leq CW_1(\epsilon,\,\rho)\norm{\theta_2-\theta_1}{\infty}. $$
\label{theo:trajectories}
\end{theorem}
\begin{proof}
Let $\varphi:\R^2\to\R$ be any Lipschitz continuous function with $\Lip{\varphi}\leq 1$, then
$$ \abs{\int_{\R^2}\varphi(y)\,d(\mu^2_t-\mu^1_t)(y)}=\abs{\theta_2(x_2(t))-\theta_1(x_1(t))}\cdot\abs{\int_{\R^2}\varphi(y)\,d(\epsilon-\rho)(y)}. $$
Taking the supremum of both sides over $\varphi$ yields
\begin{align*}
	W_1(\mu^1_t,\,\mu^2_t) &\leq \bigl(\abs{\theta_2(x_2(t))-\theta_1(x_2(t))}+\abs{\theta_1(x_2(t))-\theta_1(x_1(t))}\bigr)W_1(\epsilon,\,\rho) \\
	&\leq \bigl(\norm{\theta_2-\theta_1}{\infty}+\Lip{\theta_1}\abs{x_2(t)-x_1(t)}\bigr)W_1(\epsilon,\,\rho).
\end{align*}
An analogous result is obtained by adding and subtracting $\theta_2(x_1(t))$, but this gives $\Lip{\theta_2}$ before the second term at the right-hand side. Thus finally:
$$ W_1(\mu^1_t,\,\mu^2_t)\leq\bigl(\norm{\theta_2-\theta_1}{\infty}+\min\{\Lip{\theta_1},\,\Lip{\theta_2}\}\abs{x_2(t)-x_1(t)}\bigr)W_1(\epsilon,\,\rho) $$
and the thesis follows by plugging this in~\eqref{eq:Gronwall_x} and invoking Gronwall's inequality.
\end{proof}

Theorem~\ref{theo:trajectories} supports the numerical findings of the previous section. In both settings (i) and (ii) the purely microscopic model corresponds to $\theta_1\equiv 1$, while the hybrid model corresponds to $\theta_2=\theta_2(x(t))$ as indicated in~\eqref{eq:ell}-\eqref{eq:theta}. In general, $\norm{\theta_2-\theta_1}{\infty}=1$ as soon as $\theta_2(x)=0$ for some $x\in\R^2$, hence the relationship between the trajectories $t\mapsto x_1(t),\,x_2(t)$ depends strongly on the multiscale description of the static crowd. In setting (i) the microscopic and macroscopic distributions of the static crowd are similar, because the crowd is sparse. Consequently $W_1(\epsilon,\,\rho)$ is small and Theorem~\ref{theo:trajectories} implies that no relevant differences can be observed in the trajectories of the moving pedestrian. Conversely, in setting (ii) the two distributions of the static crowd are quite different because of the density spot. Therefore $W_1(\epsilon,\,\rho)$ is large and Theorem~\ref{theo:trajectories} admits possibly different trajectories of the moving pedestrian.

\section{Conclusions}
We have proposed a mathematical model for pedestrian movement which implements the idea of human perception as a multiscale process. In more detail, it takes into account the fact that the way in which a walker perceives and reacts to the presence of other nearby individuals changes according to various environmental factors, among which we have considered especially his/her use of space. We have modelled the perception and the consequent use of space by means of a dual micro/macroscopic representation of the nearby individuals. Our numerical and analytical results show that different types of perception greatly impact on the actual migratory paths of the walkers, which may not be reproduced by models at a single scale.


\bibliographystyle{amsplain}
\bibliography{CaSmTa-perception_crowd}

\providecommand{\bysame}{\leavevmode\hbox to3em{\hrulefill}\thinspace}
\providecommand{\MR}{\relax\ifhmode\unskip\space\fi MR }
\providecommand{\MRhref}[2]{%
  \href{http://www.ams.org/mathscinet-getitem?mr=#1}{#2}
}
\providecommand{\href}[2]{#2}
\begin{thebibliography}{1}

\bibitem{colombi2015JMB}
A.~Colombi, M.~Scianna, and A.~Tosin, \emph{Differentiated cell behavior: a
  multiscale approach using measure theory}, J. Math. Biol. \textbf{71} (2015),
  no.~5, 1049--1079.

\bibitem{cristiani2014BOOK}
E.~Cristiani, B.~Piccoli, and A.~Tosin, \emph{Multiscale {M}odeling of
  {P}edestrian {D}ynamics}, MS\&A: Modeling, Simulation and Applications,
  vol.~12, Springer International Publishing, 2014.

\bibitem{fujiyama2005WP}
T.~Fujiyama, \emph{Investigating use of space of pedestrians}, Tech. report,
  Centre for Transport Studies - University College London, January 2005.

\bibitem{fujiyama2006WP}
\bysame, \emph{Investigating density effect on the ``awareness'' area of
  pedestrians using an eye tracker}, Tech. report, Centre for Transport Studies
  - University College London, August 2006.

\bibitem{helbing2009ENCYCLOPEDIA}
D.~Helbing and A.~Johansson, \emph{Pedestrian, crowd, and evacuation dynamics},
  Encyclopedia of Complexity and Systems Science (R.~A. Meyers, ed.), vol.~16,
  Springer New York, 2009, pp.~6476--6495.

\bibitem{helbing1995PRE}
D.~Helbing and P.~Moln\'ar, \emph{Social force model for pedestrian dynamics},
  Phys. Rev. E \textbf{51} (1995), no.~5, 4282--4286.

\bibitem{seyfried2005JSM}
A.~Seyfried, B.~Steffen, W.~Klingsch, and M.~Boltes, \emph{The fundamental
  diagram of pedestrian movement revisited}, J. Stat. Mech. Theory Exp.
  \textbf{2005} (2005), P10002/1--13.

\bibitem{tosin2011NHM}
A.~Tosin and P.~Frasca, \emph{Existence and approximation of probability
  measure solutions to models of collective behaviors}, Netw. Heterog. Media
  \textbf{6} (2011), no.~3, 561--596.

\bibitem{venuti2007CRM}
F.~Venuti and L.~Bruno, \emph{An interpretative model of the pedestrian
  fundamental relation}, C. R. Mecanique \textbf{335} (2007), no.~4, 194--200.

\end{thebibliography}

\end{document}